\documentclass{elsarticle}
\usepackage{setspace} 
\usepackage[latin1]{inputenc}
\usepackage{epsfig}
\usepackage{amsmath,latexsym,amssymb,txfonts}
\usepackage{graphicx, subfigure}
\usepackage{color,caption}
\usepackage{url,amsfonts,epsfig}
\usepackage{makeidx}  

\usepackage{tikz}

\newcommand{\DENSE}{%
\setlength{\labelwidth}{12pt}%
\setlength{\labelsep}{2pt}%
\setlength{\leftmargin}{\labelwidth}%
\addtolength{\leftmargin}{\labelsep}%
\setlength{\parsep}{0pt}%
\setlength{\itemsep}{0pt}%
\setlength{\topsep}{0pt}
}
\newenvironment{ditemize}{\begin{list}%
{$-$~}{\DENSE}}{\end{list}}
\newtheorem{theorem}{Theorem}

\newtheorem{lemma}{Lemma}
\newtheorem{Corollary}{Corollary}

\newcommand{\cercle}[4]{
\node[circle,inner sep=0,minimum size={2*#2}](a) at (#1) {};
\draw[black,thick] (a.#3) arc (#3:{#3+#4}:#2);
}

\begin{document}
\begin{frontmatter}

\title{On Pairwise Compatibility of Some Graph (Super)Classes}

\author[comp]{T. Calamoneri\fnref{fn1}}
\ead{calamoneri@di.uniroma1.it}

\author[comp,INRIA]{M. Gastaldello}
\ead{mattia.gastaldello@uniroma1.it}

\author[INRIA]{B. Sinaimeri}
\ead{blerina.sinaimeri@inria.fr}

\address[comp]{Department of Computer Science, ``Sapienza'' University of Rome, Italy}
\address[INRIA]{INRIA and Universit\'e de Lyon, Universit\'e Lyon 1, LBBE, CNRS UMR558, France}

\fntext[fn1]{Corresponding Author - Partially supported by the Italian Ministry of Education,
University, and Research (MIUR) under PRIN 2012C4E3KT national
research project ``AMANDA'' -  Algorithmics for MAssive and Networked
DAta").}

%

%
%

\begin{abstract}
A graph $G=(V,E)$ is a {\em pairwise compatibility graph} (PCG) if there exists an edge-weighted tree $T$ and two non-negative real numbers $d_{min}$ and $d_{max}$ such that each leaf $u$ of $T$ is a node of $V$ and there is an edge $(u,v) \in E$ if and only if $d_{min} \leq d_T (u, v) \leq d_{max}$ where $d_T (u, v)$ is the sum of weights of the edges on the unique path from $u$ to $v$ in $T$. 
The main issue on these graphs consists in characterizing them.

In this note we prove the inclusion in the PCG class of threshold tolerance graphs and the non-inclusion of a number of intersection graphs, such as disk and grid intersection graphs, circular arc and tolerance graphs. 
The non-inclusion of some superclasses (trapezoid, permutation and rectangle intersection graphs) follows.
\end{abstract}


\end{frontmatter}

\begin{doublespace}

\section{Introduction}\label{sec:intro}

A graph $G=(V,E)$ is a {\em pairwise compatibility graph (PCG)}  if there exists an edge-weighted tree $T$ and two non-negative real numbers $d_{min}$ and $d_{max}$ such that each leaf $u$ of $T$ is a node of $V$ and there is an edge $(u,v) \in E$ if and only if $d_{min} \leq d_T (u, v) \leq d_{max}$ where $d_T (u, v)$ is the sum of weights of the edges on the unique path from $u$ to $v$ in $T$. 

This graph class arose in the context of the sampling problem on phylogenetic trees  \cite{Kal03}, subject to some biologically-motivated constraints, in order to test the reconstruction algorithms on the smaller subtrees induced by the sample. The constraints on the sample attempt to assure that the behavior of the algorithm will not be biased by the fact it is applied on the sample instead on the whole tree. 

Much attention has been dedicated to PCGs in the literature. However, many problems remain open and we are still far from a complete characterization of the PCG class.
Any progress towards the solution of the latter problem would be interesting not only from a graph theory perspective but also because it could help in the design of better sampling algorithms for phylogenetic trees.

In this note, we start from the knowledge that certain graph classes with a strong structure are known to be in PCG, and we wonder whether some of their superclasses remain in PCG.
Namely, it is well known that interval graphs are in PCG \cite{BH08}; we consider more general intersection graphs, as disk graphs, grid intersection graphs and circular arc graphs, and we prove that these three superclasses of interval graphs are not in PCG.
Permutation and trapezoid graphs, superclasses of circular arc graphs, are trivially not in PCG, too.

Moreover, we consider threshold graphs, an important subclass of interval graphs, that can be defined in terms of weights and thresholds; it is known that threshold graphs are in PCG  \cite{CPS11}; we analyze some of its superclasses that can be defined in terms of weights and thresholds, too, such as threshold tolerance and tolerance graphs; we prove here that threshold tolerance graphs still are in PCG while tolerance (i.e. parallelepiped) graphs are not.

In this way, we heavily delimitate the width of PCGs moving one step ahead toward the full comprehension of this interesting class of graphs.

\section{Superclasses of Interval Graphs}

A graph is an {\em interval} graph if it has an intersection model consisting of intervals on a straight line. 
Interval graphs are in PCG \cite{BH08}, and the witness tree is a  {\em caterpillar},  i.e. a tree in which all the nodes are within distance 1 of a central path, called spine.

In this section we consider some superclasses of interval graphs, defined as intersection graphs, and we prove that all these classes are not included in PCG, as there is at least one graph belonging to them that is not in PCG. 

A {\em disk} graph is the intersection graph of disks in the plane. 
A graph is {\em grid intersection} if it is the intersection graph of horizontal and vertical line segments in the plane.  
A {\em circular arc} graph is the intersection graph of arcs of a circle. 

It is known that graph $H$ depicted in Fig. \ref{fig.graph}.a is not in PCG \cite{DMR13}. 
On the other hand, Figures \ref{fig.graph}.b, \ref{fig.graph}.c and  \ref{fig.graph2}.a show a representation of graph $H$ as a disk graph, as a grid intersection graph  and as a circular arc graph, respectively.
This is enough to ensure the correctness of the following theorem:

\begin{theorem}
Disk graphs, grid intersection graphs and circular arc graphs are not in PCG.
\end{theorem}

\begin{figure}[t]
\begin{center}
\begin{tabular}{c  c  c}
\begin{picture}(100,80)(0,0)
\put(35,0){\circle*{6}}
\put(35,20){\circle*{6}}
\put(0,35){\circle*{6}}
\put(20,35){\circle*{6}}
\put(50,35){\circle*{6}}
\put(70,35){\circle*{6}}
\put(35,50){\circle*{6}}
\put(35,70){\circle*{6}}

\put(38,-5){$8$}
\put(33,24){$7$}
\put(-8,30){$3$}
\put(23,30){$4$}
\put(42,30){$5$}
\put(73,30){$6$}
\put(33,40){$2$}
\put(38,67){$1$}

\put(0,35){\line(1,0){20}}
\put(50,35){\line(1,0){20}}
\put(35,0){\line(0,1){20}}
\put(35,50){\line(0,1){20}}

\qbezier(0,35)(35,20)(35,20)
\qbezier(0,35)(35,0)(35,0)
\qbezier(0,35)(35,50)(35,50)
\qbezier(0,35)(35,70)(35,70)
\qbezier(20,35)(35,20)(35,20)
\qbezier(20,35)(35,0)(35,0)
\qbezier(20,35)(35,50)(35,50)
\qbezier(20,35)(35,70)(35,70)

\qbezier(50,35)(35,20)(35,20)
\qbezier(50,35)(35,0)(35,0)
\qbezier(50,35)(35,50)(35,50)
\qbezier(50,35)(35,70)(35,70)
\qbezier(70,35)(35,20)(35,20)
\qbezier(70,35)(35,0)(35,0)
\qbezier(70,35)(35,50)(35,50)
\qbezier(70,35)(35,70)(35,70)
\end{picture}
&
\begin{picture}(100,80)(0,0)
\put(25,35){\circle{30}}
\put(20,35){\circle{40}}
\put(45,15){\circle{30}}
\put(45,10){\circle{40}}
\put(45,55){\circle{30}}
\put(45,60){\circle{40}}
\put(65,35){\circle{30}}
\put(70,35){\circle{40}}

\put(42,-8){$8$}
\put(42,10){$7$}
\put(3,30){$3$}
\put(23,30){$4$}
\put(66,30){$5$}
\put(83,30){$6$}
\put(42,52){$2$}
\put(42,72){$1$}

\end{picture}
&
\begin{picture}(80,80)(0,0)

\linethickness{0.5mm}
\put(25,0){\line(0,1){80}}
\put(55,0){\line(0,1){80}}
\put(0,25){\line(1,0){80}}
\put(0,55){\line(1,0){80}}

\put(24,10){\line(0,1){60}}
\put(56,10){\line(0,1){60}}
\put(10,24){\line(1,0){60}}
\put(10,56){\line(1,0){60}}

\put(73,17){$8$}
\put(64,15){$7$}
\put(19,2){$3$}
\put(17,12){$4$}
\put(58,63){$5$}
\put(56,73){$6$}
\put(12,58){$2$}
\put(3,57){$1$}

\end{picture}\\

&\\
a&b&c
\end{tabular}
\end{center}
\caption{a. The graph $H$ that is not a PCG; b. Representation of $H$ as a disk graph; c. Representation of $H$ as a grid intersection graph.} 
\label{fig.graph}
\end{figure}
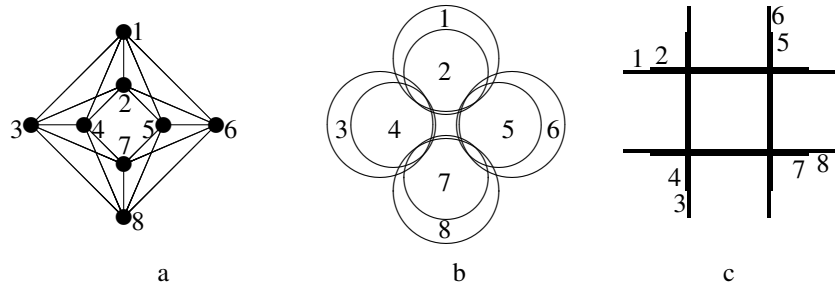

A graph is {\em rectangle (square) intersection} if it has an intersection model consisting of axis-parallel rectangular (squared) boxes in the plane. 

Rectangle (square) intersection graphs are a superclass of grid intersection graphs, and so it holds:

\begin{Corollary}
Rectangle (square) intersection graphs are not in PCG.
\end{Corollary}

Figure \ref{fig.graph2}.b shows a representation of $H$ as square intersection graph: although not necessary to prove that rectangle (square) intersection graphs are not in PCG, it will be useful for proving Theorem 2.
Moreover, we remark that rectangle intersection graphs are a superclass of Dilworth 2 graphs, that are known to be in PCG \cite{CP13}.

A {\em trapezoid} graph is the intersection graph of trapezoids between two parallel lines.
A {\em permutation} graph is the intersection graph of  straight lines between two parallels. 
The following chain of inclusions holds:

interval graphs $\subseteq$ circular arc graphs $\subseteq$ permutation graphs $\subseteq$ trapezoid graphs

\noindent
leading to the following statement:
\begin{Corollary}
Trapezoid and permutation graphs are not in PCG.
\end{Corollary}
We point out that split permutation graphs, a proper subclass of permutation graphs, is in PCG \cite{CP13}.

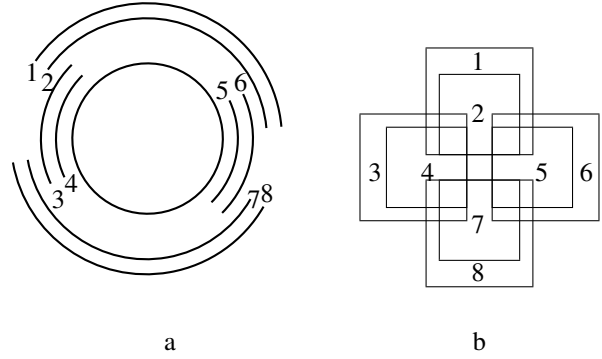
\begin{figure}[t]
\begin{center}
\begin{tabular}{c  c  c}

\begin{picture}(120,60)(0,0)

\begin{tikzpicture}
\coordinate (center) at (0,0);

\cercle{center}{1.0cm}{25}{360}

\cercle{center}{1.2cm}{25}{-70}
\cercle{center}{1.4cm}{25}{-70}
\cercle{center}{1.2cm}{205}{-70}
\cercle{center}{1.4cm}{205}{-70}
\cercle{center}{1.6cm}{145}{-140}
\cercle{center}{1.8cm}{145}{-140}
\cercle{center}{1.6cm}{190}{140}
\cercle{center}{1.8cm}{190}{140}

\put(43,-24){$8$}
\put(38,-26){$7$}
\put(33,18){$6$}
\put(26,15){$5$}
\put(-36,-26){$3$}
\put(-31,-20){$4$}
\put(-40,18){$2$}
\put(-46,22){$1$}
\end{tikzpicture}

\end{picture}
&

\begin{picture}(90,80)(0,0)
\put(10,30){\framebox(30,30)}
\put(0,25){\framebox(40,40)}
\put(30,10){\framebox(30,30)}
\put(25,0){\framebox(40,40)}
\put(30,50){\framebox(30,30)}
\put(25,50){\framebox(40,40)}
\put(50,30){\framebox(30,30)}
\put(50,25){\framebox(40,40)}

\put(42,2){$8$}
\put(42,20){$7$}
\put(3,40){$3$}
\put(23,40){$4$}
\put(66,40){$5$}
\put(83,40){$6$}
\put(42,62){$2$}
\put(42,82){$1$}

\end{picture}\\
&\\
a&b
\end{tabular}
\end{center}
\caption{a. Representation of $H$ as circular arc graph; b. Representation of $H$ as grid intersection graph.} 
\label{fig.graph2}
\end{figure}

A graph is a \emph{tolerance graph} \cite{GMT84} if to every node $v$ can be assigned a closed interval $I_v$ on the real line and a tolerance $t_v$ such that $x$ and 
$y$ are adjacent if and only if $|I_x \cap I_y| \geq \min\{t_x , t_y\}$, where $|I|$ is the length of the interval $I$. 
Tolerance graphs can be described through another intersection model, as they are equivalent to {\em parallelepiped} graphs, defined as the intersection graphs of {\em special parallelepiped} on two parallel lines.
Let $L$ and $M$ be two parallel lines in 3-dimensional Euclidean space: 
$L=\{(\cdot,0,0)\}$ and $M=\{(\cdot,1,0)\}$; 
a {\em special parallelepiped} on two parallel lines is either
\begin{itemize}
\item the convex hull of eight points:
$A = (a,1,0)$; $	A' = (a,1,z)$;
$B = (b,0,0)$; $B' = (b,0,z)$;
$C = (c,0,0)$; $C' = (c,0,z)$;
$D = (d,1,0)$; $D' = (d,1,z)$
for some $a,b,c,d,z \geq 0$ with $ADBC$ a parallelogram between L and M; or
\item the line segment between $(a,1,z)$ and $(b,0,z)$ for some $a,b,z \geq 0$.
\end{itemize}

Consider now the squares depicted in Figure \ref{fig.graph2}.b and consider them as the bases of cubes lying in the 3-dimensional semispace with not negative $z$ values. 
The intersection graph of these cubes is $H$, so showing that it is a tolerance graph, and proving the following theorem:

\begin{theorem} 
Tolerance graphs are not in PCG.
\end{theorem}

\section{Superclasses of Threshold Graphs}

A graph $G=(V,E)$ is a \emph{threshold} graph  if there is a real number $t$  and for every node $v$ in $V$ there is a real weight $a_v$ such that: $\{v,w\}$ is an edge if and only if  $a_v + a_w \geq  t$ \cite{MRT88,MP95}. They are in PCG, as shown in \cite{CPS13}, and the witness tree is a star.

We consider a superclass of threshold graphs that can be also defined in terms of weights and thresholds: threshold tolerance graphs; we put them in relation with the PCG class proving that threshold tolerance graphs are in PCG.

A graph $G=(V,E)$ is a \emph{threshold tolerance} graph if it is possible to associate weights and tolerances with each node of $G$ so that two nodes are adjacent exactly when the sum of their weights exceeds either of their tolerances.
More formally, there are positive real-valued functions, weights $g$ and tolerances $t$ on $V$ such that $\{x, y\} \in E$  if and only if $g(x)+g(y) \geq  \min{(t(x), t(y))}$.
In the following we  indicate with $G=(V,E, g, t)$ a graph in this class.
Threshold tolerance graphs have been introduced in \cite{MRT86} as a generalization of threshold graphs, that can be obtained by defining the tolerance function as a constant \cite{MRT88}. 

Before proving that threshold tolerance graphs are in PCG, we need to demonstrate a preliminary lemma stating that, when we deal with threshold tolerance graphs, w.l.o.g. we can restrict ourselves to the case when $g$ and $t$ take only positive integer values.

\begin{lemma}
\label{lemma.int}
A graph $G=(V,E)$ is a threshold tolerance if and only if there exist two functions $g,t:V\rightarrow  \mathbb{N}^+$ such that $(V,E,g,t)$ is threshold tolerance.
\end{lemma}

\begin{proof}
Clearly if $f,g$ exist then by definition $G$ is a threshold tolerance graph. Suppose now $G$ is a threshold tolerance graph which weight and tolerance functions $g$ and $t$ are both defined from $V$ to $\mathbb{R}^+$. 
Nevertheless, it is not restrictive to assume that $g,t: V \rightarrow \mathbb{Q}^+$ in view of the density of rational numbers among real numbers.
So, we can assume that, for each $v \in V$, $t(v)=n_v / d_v$ and $g(v)=n'_v / d'_v$.
Let $m$ be the minimum common multiple of all the numbers $d_v$ and $d'_v$, $v \in V$. 
So we can express $t(v)$ and $g(v)$ as $t(v)=\frac{n_v \cdot m/d_v}{m}$ and 
$g(v)=\frac{n'_v \cdot m/d'_v}{m}$ where $m/d_v$ and $m/d'_v$ are integer values.

Define now the new functions $\tilde{g}$ and $\tilde{t}$ as $\tilde{g}(v)=g(v) \cdot m$ and $\tilde{t}(v)=t(v) \cdot m$, $v \in V$.
Clearly, it holds that $\tilde{g}: V \rightarrow \mathbb{Q}^+$ while $\tilde{t}: V \rightarrow \mathbb{N}^+$.

In order to prove the claim, it remains to prove that $\tilde{g}$ and $\tilde{t}$ define the same graph defined by $t$ and $g$.
This descends from the fact that $\tilde{g}(x)+\tilde{g}(y)=(g(x)+g(y)) \cdot m \geq \mbox{min}(t(x), t(y)) \cdot m= \mbox{min}(\tilde{t}(x), \tilde{t}(y))$ if and only if $g(x)+g(y) \geq \mbox{min}(t(x), t(y))$. \qed
\end{proof}

\begin{theorem}
\label{th.TTinmLPG}
Threshold tolerance graphs are in PCG.
\end{theorem}
\begin{proof}
Let $G=(V,E, g, t)$ be a threshold tolerance graph. 
Let $K=\max_v{t(v)}$.
In view of Lemma \ref{lemma.int}, it is not restrictive to assume that $g: V \rightarrow \mathbb{N}^+$, so we split the nodes of $G$  in groups $S_1, \ldots, S_K$ such that $S_i=\{v \in V(G): t(v)=i\}$. Observe that for some values of $i$ the set $S_i$  can be empty.  

We associate to $G$ a caterpillar $T$ as in Figure \ref{fig:tree}. 
The spine of the caterpillar is formed by $K$ nodes, $x_1, \ldots, x_K$, and each node $x_i$ is connected to the leaves $l_v$ corresponding to nodes $v$ in $S_i$. 
The weights $w$ of the edges of $T$ are defined as follows: 
\begin{ditemize}
\item For each edge of the spine $w(x_i,x_{i+1})=0.5$ for $0 \leq i \leq K-1$.  
\item For each leaf $l_v$ connected to the spine through node $x_i$ we assign a weight $w(v, x_i)=g(v)+\frac{K-t(v)}{2}$. 
\end{ditemize}


\begin{figure}[h]
\begin{center}
\includegraphics[scale=0.5]{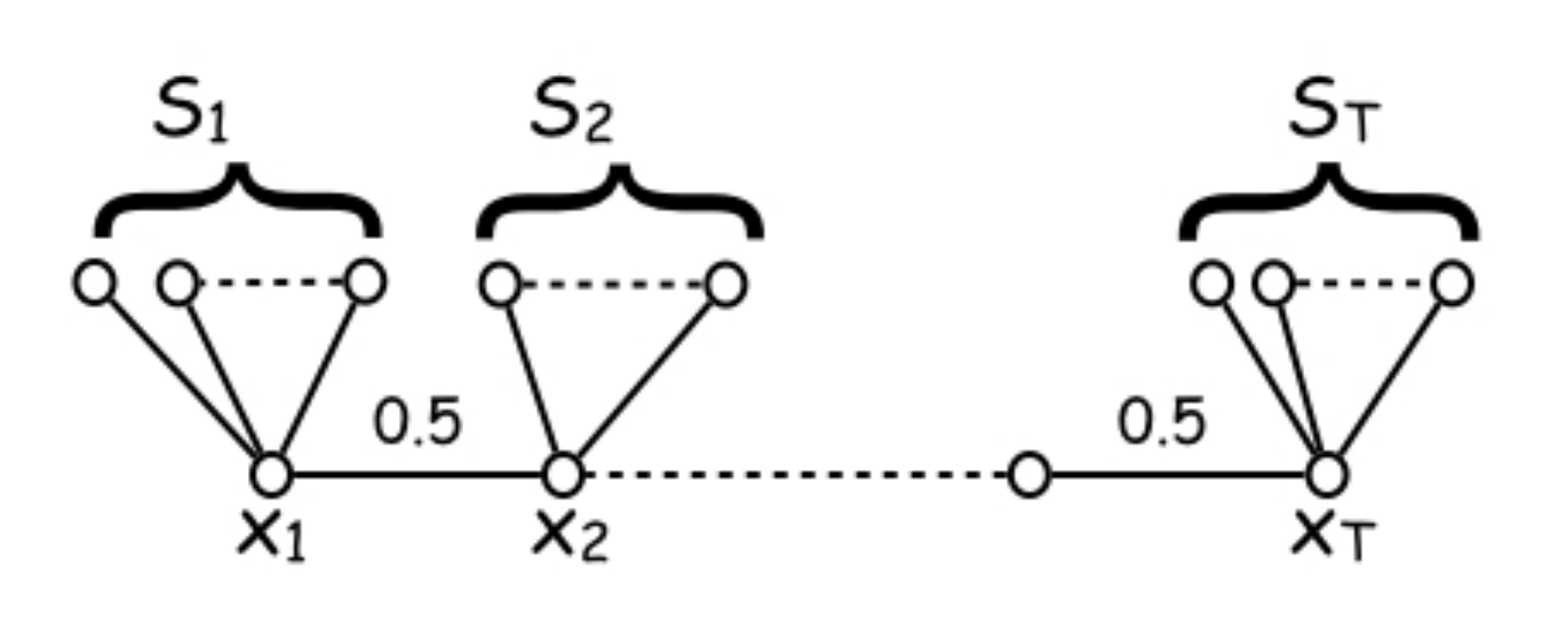}
\caption{\footnotesize{The caterpillar used in the proof of Theorem \ref{th.TTinmLPG} to prove that threshold tolerance graphs are in mLPG.}}
\label{fig:tree}
\end{center}
\end{figure}

We show that $G=PCG(T,w, 2, K)$. 
To this purpose consider two nodes $u$ and $v$ in $G$. 
By construction, in $T$ we have that $l_u$ is connected to $x_{t(u)}$ and $l_v$ to $x_{t(v)}$, where $t(u)$ and $t(v)$ are not necessary distinct. 
Clearly, w.l.o.g we can assume $t(v) \geq t(u)$, i.e. $t(u)= \min{(t(u), t(v))}$. 
We have that 
\begin{flalign}
\notag
d_T(l_u, l_v)&=w(l_u,x_{t(u)})+\frac{t(v)-t(u)}{2} + w(l_v,x_{t(v)}) \\ \notag
&=g(u)+\frac{K-t(u)}{2} + \frac{t(v)-t(u)}{2}+ g(v)+\frac{K-t(v)}{2} \\ \notag
&= g(u) + g(v) + K - t(u)
\end{flalign}

Clearly, $d_T(l_u, l_v) \geq 2$ since $g(v)$ is a positive integer, due to Lemma \ref{lemma.int}; moreover $d_T(l_u, l_v) \geq K$ if and only if  $g(u)+g(v) \geq  t(u)=\min{(t(u), t(v))}$ and this proves the assertion. \qed
\end{proof}

\section{Conclusions}

In this note, we observe that two important classes of graphs, i.e. interval and threshold graphs, have a strong structure and are known to be in PCG.
We consider some superclasses of these classes, obtained by generalizing their definitions (either intersection graphs or graphs defined in terms of weights and thresholds) and we studied them w.r.t. the class of Pairwise Compatibility Graphs.
In this way, we have moved one step ahead toward the full comprehension of the interesting class of PCGs.

\end{doublespace}



\end{document}